\documentclass[10pt]{iopart}
\usepackage{iopams}  
\usepackage{setstack}
\usepackage[numbers,sort&compress]{natbib}
\pdfoutput=1
\usepackage[utf8]{inputenc}
\usepackage[english]{babel}
\usepackage[T1]{fontenc}
\usepackage{hyperref}
\usepackage{setspace}
\usepackage{braket}
\usepackage{dsfont}
\usepackage{refcount}

\usepackage{tikz}
\usepackage{pgfplots}

\usepackage{xparse}

\usepgfplotslibrary{fillbetween}

\usepackage{dsfont}

\usepackage{graphicx}

\newtheorem{definition}{Definition}
\newtheorem{lem}[definition]{Lemma}
\newtheorem{prop}[definition]{Proposition}
\newtheorem{thm}[definition]{Theorem}
\newtheorem{cor}[definition]{Corollary}
\newtheorem{conjecture}[definition]{Conjecture}

\newenvironment {proof}{\textit{Proof.}}{\hfill \ensuremath{\square}\vspace{1em}}

\global\long\def\text#1{{\rm #1}}

\global\long\def\binom#1#2{{#1 \choose #2}}

\renewcommand{\substack}[2]{{#1 \atop #2}}

\global\long\def\one{\mathds{1}}
\global\long\def\trace{\Tr}
\global\long\def\diag{{\rm diag}}
\global\long\def\ketbra#1#2{\ket{#1}\!\bra{#2}}

\global\long\def\X{X}
\global\long\def\Y{\hspace{0.1em}Y}
\global\long\def\Z{Z\hspace{0.13em}}

\pgfplotsset{compat=1.14}

\begin{document}

\title{Characterizing quantum states via sector lengths}

\author{N Wyderka and O Gühne}

\address{Naturwissenschaftlich-Technische Fakult{\"a}t, 
Universit{\"a}t Siegen, Walter-Flex-Stra{\ss}e~3, 57068 Siegen, 
Germany}

\begin{abstract}
Correlations in multiparticle systems are constrained by restrictions from 
quantum mechanics. A prominent example for these restrictions are monogamy 
relations, limiting the amount of entanglement between pairs of particles in
a three-particle system. A powerful tool to study correlation constraints 
is the notion of sector lengths. These quantify, for different $k$, the 
amount of $k$-partite correlations in a quantum state in a basis-independent
manner. We derive tight bounds on the sector lengths in multi-qubit states 
and highlight applications of these bounds to entanglement detection, monogamy 
relations and the $n$-representability problem. For the case of two- and
three qubits we characterize the possible sector lengths completely and
prove a symmetrized version of strong subadditivity for the linear entropy.
\end{abstract}

\noindent{\it Keywords\/}: correlations, sector length, linear entropy, monogamy, strong subadditivity
\submitto{\jpa}

\section{Introduction}
Correlations between particles are central for many physical phenomena, ranging from
phase transitions in condensed matter systems to applications like quantum metrology.
These non-local correlations, however, cannot be completely arbitrary as they are 
subject to restrictions from quantum mechanics. A prominent example for these kind 
of restrictions concerns entanglement in three-partite systems.  Here, a monogamy 
relation known as the Coffman-Kundu-Wootters-inequality limits the sum of the 
entanglement between the first and the second party and the entanglement 
between the first and the third party~\cite{Coffman2000}.

A useful concept to describe the correlation structure of quantum states is
the so-called sector length \cite{aschauer2003local}. Sector lengths for $n$-party
quantum states are quadratic expressions and quantify, for different $k\leq n$, 
the amount of $k$-partite correlations in the state. Thus, to any $n$-qubit state 
one assigns a tuple $(A_1,\ldots A_n)$ of  sector lengths and infers 
properties of the state based on the sector length configuration. Sector lengths
are, as all correlation measures, invariant under local unitary transformations 
\cite{Gingrich2002}.
They are expressible in terms of purities of the reduced states of a system, and 
as such, they can be experimentally characterized by randomized measurements on 
a single copy of the state \cite{brydges2019probing}.

Consequently, sector lengths have been used for many purposes, for example, entanglement 
detection \cite{klockl2015characterizing}, deriving monogamy relations \cite{eltschka2015monogamy} and excluding the existence of certain 
absolutely maximally entangled states \cite{huber2017absolutely}. In the context of quantum coding theory, sector lengths are known as weight enumerator theory and are used to characterize quantum codes \cite{gottesman1997stabilizer}. Furthermore, 
bounds on $k$-sector lengths with $k<n$ can be used to find necessary conditions 
for a set of reduced density matrices of up to $k$ of the parties to be 
compatible with a global state. This problem is known as the representability problem  \cite{higuchi2003one, klyachko04, Klyachko2006}. 

Finally, the purity of a state is related to its coherence and serves as a complementary quantity \cite{streltsov2017colloquium,cheng2015complementarity,giorda2017two,streltsov2018maximal,kumar2017quantum}. As sector lengths are in one-to-one correspondence with purities of subsystems of a state, bounds on sector lengths can be used to shed light on coherence structures in multipartite quantum states.

In this paper we first find exact bounds on individual sectors $A_k$ for 
$k\in\{2,3,n\}$. Furthermore, we fully classify the set of admissible 
tuples of sector lengths for two- and three-qubit states by characterizing all 
bounds on linear combinations of the sector lengths. Interestingly, we show that in these cases,
the admissible sector lengths form a convex polytope that can be characterized
by few constraints. One of these constraints can be viewed as a symmetrized 
version of strong subadditivity (SSA) of the linear entropy.

The paper is structured as follows: First, we will define sector lengths and review 
known relations between them. Then, we find tight bounds on the individual sectors 
$A_2,A_3$ and $A_n$ in $n$-qubit states. There, we highlight connections to 
monogamy of entanglement and apply our results to the representability problem  
and the problem of entanglement detection. Next, we extensively study the cases 
of two and three qubits. To that end, we describe how to translate between sector 
lengths, linear entropies and mutual linear entropies, which are in one-to-one correspondence. 
We completely characterize the allowed sector length configurations by 
considering a symmetrized SSA for linear entropies for three-qubit systems. 
While it is known that SSA does not hold in general for the linear entropy 
\cite{Petz2015}, we show, using techniques from semidefinite programming (SDP),
that the symmetrized version is true for three qubits. 

\section{Basic definitions}

Consider a quantum state $\rho$ of $n$ qubits. We expand the state in terms of the Bloch basis, i.e., in terms of tensor products of Pauli matrices and group the terms according to the number of non-trivial Pauli matrices in each term, i.e.,
\begin{eqnarray} \fl
\phantom{Ig}\rho &= \frac{1}{2^{n}}\left( \one+\sum_{i=1}^{n}\sum_{a\in\left\{ x,y,z\right\} }\alpha_{i,a}\sigma_{a}^{(i)} + \sum_{i<j=1}^{n}\sum_{a,b\in\left\{ x,y,z\right\} }\beta_{ij,ab}\sigma_{a}^{(i)}\sigma_{b}^{(j)}+\ldots\right) \label{eq:bloch-1} \nonumber \\ \fl
 &=  \frac{1}{2^{n}}(\one+P_{1}+P_{2}+\ldots+P_{n}),
\end{eqnarray}
where $\sigma_{a}^{(i)}$ denotes the Pauli operator acting on particle
$i$ in direction $a\in \{x,y,z\}$, padded with identities on the other particles. 
We group the terms by the number of non-trivial $\sigma$-matrices
and call the sum of all terms with $k$ matrices $P_{k}$. 
As $\rho$ is hermitian, the coefficients $\alpha,\beta,\ldots$ are
real and the $P_k$ are hermitian, as well. Note that the only term that is not traceless is the unit operator,
thus the normalization $2^{-n}$ is chosen such that $\trace(\rho)=1$.

As an example, consider the Greenberger-Horne-Zeilinger (GHZ) state of three particles, $\ket{\text{GHZ}}=\frac{1}{\sqrt{2}}(\ket{000}+\ket{111})$. In terms of Pauli operators, the density matrix reads
\begin{eqnarray} \fl
\phantom{Ig}\rho_{\text{GHZ}} =  \frac{1}{2^3}( \one\one\one+ZZ\one+Z\one Z+\one ZZ +XXX-XYY-YXY-YYX).
\end{eqnarray}
Here and in the following, we skip the tensor product symbol for better readability and write $X$, $Y$ and $Z$ for the Pauli matrices $\sigma_x, \sigma_y$ and $\sigma_z$. Thus, $\one ZZ$ means $\one\otimes \sigma_z \otimes \sigma_z$.
In this example, $P_1 = 0$, $P_2 = ZZ\one+Z\one Z+\one ZZ$ and $P_3 = XXX-XYY-YXY-YYX$.

The sector length $A_{k}$ captures the amount of $k$-body correlations in a state. It is defined as \cite{aschauer2003local}
\begin{eqnarray}
A_{k}(\rho) := &\frac{1}{2^{n}}\trace[P_{k}(\rho)^{2}] 
             = \sum_{\Xi_k}\trace[\Xi_k\rho]^{2},
\end{eqnarray}
where the sum spans over all Pauli operators $\Xi_k$ acting on $k$ of the parties nontrivially.
Using the expansion (\ref{eq:bloch-1}), this means that $A_1 = \sum_{i=1}^n \sum_{a\in\{x,y,z\}} \alpha_{i,a}^2$ is the sum of the squares of the local Bloch vectors, $A_2 = \sum_{i<j} \sum_{a,b\{x,y,z\}} \beta_{ij,ab}^2$, etc.  Note that $A_{0}=1$ by normalization. As an example, the GHZ state above has sector length configuration $(A_1,A_2,A_3) = (0,3,4)$. We stress that while we used an explicit choice of a basis to define the $A_i$ they are invariant under local unitary operations, and as such, they are independent of the choice of the local basis.

Considering the set of all quantum states of $n$ parties, we are
interested in the tuples $(A_{1},\ldots,A_{n})$ that are
attainable. First, we find tight bounds on the individual sectors.
These bounds can always be realized
by pure states, as the quantity $A_{i}$ is convex: $A_{i}(\rho)\leq\sum_{j}p_{j}A_{i}(\ket{\psi_{j}})$
if $\rho=\sum_{j}p_{j}\ketbra{\psi_{j}}{\psi_{j}}$.
Thus, we start by listing some basic facts about sector lengths of pure states. In this case, $\rho=\rho^{2}$ and therefore $\sum_{k=0}^{n}A_{k}=2^{n}$. In fact, the sum of all sector lengths is equal to the purity of the state up to a factor of $2^n$.

Additionally, there are many relations among the $A_{i}$ for pure states: Choosing  $\rho=\ketbra{\psi}{\psi}$ and a subsystem $S\subset\{1,\ldots,n\}$, one can define the reduced state of particles $S$, $\rho_S := \trace_{\bar{S}}(\rho)$, where $\bar{S}=\{1,\ldots,n\}\setminus S$. Using the Schmidt decomposition,
one can show that $\trace[\rho(\rho_{S}\otimes \one_{\bar{S}})]=\trace[\rho(\one_{S} \otimes \rho_{\bar{S}} )]$.
Summing this identity over all subsets of size $m\leq n$ yields an equation for pure states that is expressible in terms of sector lengths \cite{huber2018ulam}:
\begin{eqnarray}
M_m := {} 2^{m}\sum_{j=0}^{n-m}&\binom{n-j}{m}A_{j}-  2^{n-m}\sum_{j=0}^{m}\binom{n-j}{n-m}A_{j}=0\label{eq:schmidt}
\end{eqnarray}
for all integer $0\leq m\leq n$, where for $m=0$ one obtains the purity equality, $\sum_i A_i = 2^n$. The relations $M_m = 0$ are known in the more general context of coding theory as MacWilliams' identities \cite{macwilliams1977theory}. A subset of $\left\lceil \frac{n}{2}\right\rceil$ of them are linearly independent equations and allows for the elimination
of certain $A_{i}$ if the state is known to be pure.

\section{Bounds on individual sector lengths}

We start by proving some bounds on the smallest sector lengths. First of all, it is known that
\begin{eqnarray}\label{eq:bounda1}
    A_1 \leq n
\end{eqnarray} for $n$-qubit states, which is attained for pure product states like $\ket{0\ldots0}$.
This is because $A_1(\rho)$ is given by the sum of all $A_1(\rho_i)$ of the one-party reduced states $\rho_i$ of $\rho$, corresponding to the squared magnitude of the Bloch vector, which is bounded by one.

\subsection{\texorpdfstring{Bounds on $A_2$}{Bounds on A2}} \label{sec:boundsa2}

While the bound (\ref{eq:bounda1}) is trivial, the tight bounds on $A_{2}$ are only known for
$n=2$ and $n=3$ so far. For $n=2$, the bound is given by $A_{2}\leq3$, as for the purity holds $\trace(\rho^2) = 2^{-2}[1+A_1(\rho)+A_2(\rho)]\leq 1$. For $n=3$, however, we obtain from
$M_1 = 0$ in (\ref{eq:schmidt}) for pure states that $A_{2}=3$, and therefore by convexity for all states $A_2\leq 3$. We will show here that for $n\geq3$, the bound is given by $A_{2}\leq\binom{n}{2}$, using the following Lemma.
\begin{lem}
\label{lem:l1}If for all quantum states $\rho$ of $n_{0}$ qubits
it holds that $A_{k}(\rho)\leq\binom{n_{0}}{k}$,
then for all states $\rho^{\prime}$ of $n\geq n_{0}$ qubits, it holds that $A_{k}(\rho^{\prime})\leq\binom{n}{k}$.
\end{lem}
\begin{proof}
We prove the Lemma by induction over the number of qubits $n$. 
Let the statement be true for a fixed $n\geq n_{0}$ and consider a state $\rho$ of $n+1$
parties. There are $n+1$ different $n$-party marginal states of $\rho$, $\rho_{\bar{j}}:=\trace_{j}(\rho)$ for $j\in\{1,\ldots,n+1\}$. For each of them it holds
by assumption that $A_{k}(\rho_{\bar{j}})\leq\binom{n}{k}$. 

Every $k$-body correlation among the parties $i_1,\ldots,i_k$ that is present in $\rho$ is also present in the reduced states that contain the parties $i_1,\ldots,i_k$. This is the case for $(n+1-k)$ of the $(n+1)$ different reductions.
Thus,
\begin{eqnarray}
    \sum_{j=1}^{n+1} A_k(\rho_{\bar{j}}) = (n+1-k)A_k(\rho).
\end{eqnarray}
The left hand side of this equation is bounded by assumption by $(n+1)\binom{n}{k}$, thus we have that
\begin{eqnarray}
A_{k}(\rho)\leq\frac{n+1}{n+1-k}\binom{n}{k}=\binom{n+1}{k}.
\end{eqnarray}
\end{proof}
\begin{prop}
\label{prop:c1}For all qubit states of $n\geq3$ parties, it holds that
$A_{2}\leq\binom{n}{2}$. The bound is tight.
\end{prop}
\begin{proof}
For $n=3$, from $M_1 = 0$ in (\ref{eq:schmidt}) we have that $A_{2}=3=\binom{3}{2}$.
Thus, Lemma~\ref{lem:l1} applies and therefore $A_{2}\leq\binom{n}{2}$
for all $n$-qubit states with $n\geq3$. 

Concerning the tightness, consider the pure product state $\ket{0\ldots0}\bra{0\ldots0}=(\frac{\one+Z}{2})^{\otimes n}$.
It has weights given by $(A_{1},A_{2},\ldots,A_{n})=\left(\binom{n}{1},\binom{n}{2},\ldots,\binom{n}{n}\right)$
and reaches the bound.
\end{proof}

Note that in \cite{markiewicz2013detecting} the authors prove a weaker statement of Proposition~\ref{prop:c1} for the sum of all bipartite
correlation terms involving $X$ and $Y$ only, for which the same bound is obtained.

Using the same induction technique and the base case of four qubits, we can prove an even stronger, non-symmetric version of Proposition~\ref{prop:c1} for $n\geq4$, by summing only those contributions to $A_2$ that involve correlations with the (arbitrarily chosen) first qubit.
\begin{prop}\label{prop:c2}
For all qubit states of $n\geq4$ parties, it holds that
$\sum_{j=2}^{n}A_{2}(\rho_{1j})\leq n-1$.
\end{prop}
\noindent For the proof, see Appendix A.

Proposition~\ref{prop:c2} states that in a multi-qubit state, the bipartite correlations of a party with any of the other parties, on average cannot exceed one. Note that maximally entangled bipartite reduced states would obey $A_2=3$, and separable two-qubit states obey $A_2\leq 1$. Thus, Propositions~\ref{prop:c1} and \ref{prop:c2} can be seen as monogamy relations limiting the shared entanglement between a party with the rest, and Proposition~\ref{prop:c2} is in close connection to the Osborne-Verstraete relation \cite{Osborne2006}.

Furthermore, these bounds are useful in the context of the $2$-representability problem  \cite{higuchi2003one, klyachko04, Klyachko2006}. There, one wants to decide whether a set of two-body marginals is compatible with a common global state. While the $1$-representability problem for qubits is solved (i.e., the same problem with a set of one-body marginals) \cite{higuchi2003one} and yields a polytope of compatible eigenvalues, the $k$-representability problem for $k>1$ is in general hard to decide \cite{liu2007quantum}. 
However, Proposition \ref{prop:c2} can be turned into a set of necessary conditions on the spectra of a set of two-body marginals in order to be compatible:
\begin{cor}\label{cor:crepresent}
Let $\{\rho_{ij}\}_{1\leq i<j\leq n}$ denote a set of two-qubit states with eigenvalues $\lambda_k^{(ij)}$. Let their compatible one-qubit marginals be denoted by $\{\rho_i\}_{1\leq i\leq n}$ with spectra $\lambda_k^{(i)}$. If they originate from a common global state, then for the spectra of the matrices it holds that for all $1\leq i \leq n$:
\begin{eqnarray}
    2\sum_{j\neq i} \sum_{k=1}^4 (\lambda_k^{(ij)})^2 \leq \sum_{j\neq i} \sum_{k=1}^2 (\lambda_k^{(j)})^2 + (n-1)\sum_{k=1}^2 (\lambda_k^{(i)})^2.
\end{eqnarray}
\end{cor}
\begin{proof}
Note that for an $n$-qubit state $\rho$, $\trace(\rho^2) = \sum_{k=1}^{2^n} \lambda_k^2$, where $\lambda_k$  are the eigenvalues of $\rho$. Additionally, for the two-body marginal $\rho_{ij}$, the purity is given by $\trace(\rho_{ij}^2) = \frac14(1+A_1^{(i)}+A_1^{(j)}+A_2^{(ij)})$.

This allows to write $A_2^{(ij)}$ as a function of purities and thus as a function of eigenvalues, i.e.~
\begin{eqnarray}
    A_2^{(ij)} & = 4\sum_{k=1}^4 (\lambda_k^{(ij)})^2 - 2\sum_{k=1}^2 [(\lambda_k^{(i)})^2 + (\lambda_k^{(j)})^2]+1,
\end{eqnarray}
where $\lambda_k^{(ij)}$ are the eigenvalues of $\rho_{ij}$ and $\lambda_k^{(i)}, \lambda_k^{(j)}$ the eigenvalues of $\rho_{i},\rho_{j}$, respectively. Then for each fixed choice of $i$, the claim follows by using $\sum_{j=1,j\neq i}^n A_2^{(ij)} \leq n-1$ from Proposition \ref{prop:c2}.
\end{proof}

\subsection{\texorpdfstring{Bounds on $A_3$ and higher sectors}{Bounds on A3 and higher sectors}} \label{sec:boundsa3}

Up to here, the results involved two-body correlations only. In this
section, we generalize some of the statements to three-body correlations
and the sector length $A_{3}$. Recalling the statement of Lemma~\ref{lem:l1},
we know that if for some $n_{0}\geq3$, $A_{3}(\rho)\leq\binom{n_{0}}{3}$
for all $\rho$ of $n_{0}$ qubits, then the same bound holds for
all $n>n_{0}$ as well. The question arises whether such an $n_{0}$
exists. For $n=3$, $A_{3}(\ket{\text{GHZ}})=4>\binom{3}{3}=1$. For $n=4$, there
exist states with $A_{3}=8>\binom{4}{3}=4$, for example the highly entangled state  \cite{kraus2003entanglement, guhne2009entanglement, Osterloh2006}
\begin{eqnarray}\label{eq:5chi}
    \ket{\chi} = \sqrt{6}^{-1}(\ket{0001}+\ket{0010}+\ket{0100}+\ket{1000} + \sqrt2\ket{1111}).
\end{eqnarray}
Nevertheless, for $n\geq5$ the bound holds. To show this, we need to introduce
an additional technique, namely the so-called shadow inequalities \cite{rains2000polynomial}.

Let $M$ and $N$ be two positive semidefinite hermitian operators
acting on an $n$-particle space. Then for all $T\subset\{1,\ldots,n\}$
\cite{rains2000polynomial, eltschka2018exponentially},
\begin{eqnarray}\label{eq:preshadow}
\sum_{S\subset\{1,\ldots,n\}}(-1)^{\vert S\cap \bar T\vert}\trace[\trace_{\bar S}(M)\trace_{\bar S}(N)]\geq0.
\end{eqnarray}

\begin{table}[t]
\caption{Translation of the various sector bounds into inequalities for linear entropy and mutual
entropy.\label{tab:translation}}
\begin{tabular}{c||c|c|}
\br
Origin & Eq.~(\ref{eq:bounda1})  & Proposition~\ref{prop:c1}\tabularnewline
Min.~$n$ & $n \geq 1$ & $n\geq 3$ \tabularnewline

 Sector len.& $A_{1}\leq n$ & $A_{2}\leq\binom{n}{2}$  \tabularnewline

 Lin.~ent.& $S_\text{L}^{(1)}\geq0$ & $S_\text{L}^{(2)}\geq\frac{n-1}{2}S_\text{L}^{(1)}$\tabularnewline

 Mut.~ent.& $I_\text{L}^{(1)}\geq0$ & $I_\text{L}^{(2)}\leq\frac{n-1}{2}I_\text{L}^{(1)}$ \tabularnewline
 \br
\br
Origin & Proposition~\ref{prop:a3bound} &
Corollary~\ref{prop:i3bound}\tabularnewline
Min.~$n$ & $n \geq 5$ & $n \geq 3$\tabularnewline

 Sector len. & $A_{3}\leq\binom{n}{3}$ & $\binom{n}{3} + A_3 \geq \frac13\binom{n-1}{2}A_1 + \frac{n-2}{3}A_2$ \tabularnewline

 Lin.~ent.& $S_\text{L}^{(3)}\geq \frac{n-2}{2}S_\text{L}^{(2)} -\frac{1}{4}\binom{n-1}{2}S_\text{L}^{(1)}$ &
 $S_\text{L}^{(3)} \leq \frac{n-2}{3} S_\text{L}^{(2)} - \frac13\binom{n-1}{2}S_\text{L}^{(1)}$\tabularnewline

 Mut.~ent.& $I_\text{L}^{(3)}\geq \frac{n-2}{2}I_\text{L}^{(2)} -\frac{1}{4}\binom{n-1}{2}I_\text{L}^{(1)}$ &
 $I_\text{L}^{(3)}\leq \frac{n-2}{3} I_\text{L}^{(2)}$\tabularnewline
 \br
\end{tabular}
\end{table}

Here, $\bar S=\{1,\ldots,n\}\setminus S$ and $\trace_{\bar S}$ denotes
the partial trace of systems $\bar S$.

Summing over all $T$ with $\vert T\vert=k$
yields a set of inequalities $B_{k}\geq 0$:
\begin{eqnarray}
B_{k}:=\!\!\!\!\!\!\!\!\!\sum_{\substack{T,S\subset\{1,\ldots,n\},}{\vert T\vert=k}}\!\!\!\!\!\!(-1)^{\vert S\cap \bar T\vert}\trace[\trace_{\bar S}(M)\trace_{\bar S}(N)]\geq0.
\end{eqnarray}
Choosing $M=N=\rho$, the right-hand side can be evaluated in terms of the sector lengths
to read \cite{calderbank1998quantum,rains1999quantum}
\begin{equation}
B_{k}=\frac{1}{2^{n}}\sum_{r=0}^{n}(-1)^{r}K_{k}(r;n)A_{r}\geq0\label{eq:shadowineq}
\end{equation}
with the Kravchuk polynomials 
\begin{eqnarray}
K_{k}(r;n)=\sum_{j=0}^{k}(-1)^{j}3^{k-j}\binom{r}{j}\binom{n-r}{k-j}.
\end{eqnarray}
For $k=0$, $B_0 = \frac{1}{2^n}[\sum_{j=0}^n (-1)^j A_j] \geq 0$ which is known in the context of state inversion \cite{hall2005multipartite}. Note that other references denote the inequalities $B_k$ by $S_k$. Here, we chose $B_k$ instead in order to avoid confusion with the linear entropy.
Using these inequalities, we are in position to prove the following bound:
\begin{prop}\label{prop:a3bound}
For all qubit states of $n\geq 5$, it holds that $A_{3}\leq\binom{n}{3}$. For $n=3$, the bound is given by $A_3 \leq 4$; for $n=4$, it is given by $A_3 \leq 8$. The bounds are tight.
\end{prop}

\begin{proof}
For $n=3$ and $n=4$, we use a linear program that involves the purity $M_0 = 0$ from (\ref{eq:schmidt}) and state inversion inequality $B_0 \geq 0$. For $n=3$, these two equations read
\begin{eqnarray}
    1+A_1 + A_2 + A_3 &= 8, \\
    1-A_1 + A_2 - A_3 &\geq 0.
\end{eqnarray}
Subtracting the second inequality from the first and using $A_1\geq 0$, we obtain $A_3\leq 4$. The same  works for $n=4$.

For $n\geq 5$, we prove the statement for $n=5$. By use of Lemma~\ref{lem:l1}, the result
will then be true for larger $n$ as well.
We can assume that the total state is pure, as convex combinations
of pure states will never increase any sector length.
Using a linear program involving relations $M_j = 0$ for $j\in\{0,1,2\}$ from (\ref{eq:schmidt}), $B_{1} \geq 0$ reduces to $A_{3}\leq10=\binom{5}{3}$.

Concerning the tightness, consider the GHZ state for $n=3$ having $A_3 = 4$ and the state $\ket{\chi}$ for $n=4$, given in Eq.~(\ref{eq:5chi}). For $n\geq5$, consider any product state like $\ket{0}^{\otimes n}$ with sector lengths $A_k = \binom{n}{k}$.
\end{proof}

Numerically, a similar statement seems to hold for $A_{4}$ for states
of at least $8$ qubits, but using a linear program, one can show that shadow inequalities are insufficient to show it. Still, we conjecture:
\begin{conjecture}
\label{conj:conj1}For all $k$ there exists an $n_{0}$, such that
for all $n\geq n_{0}$, $A_{k}\leq\binom{n}{k}$ holds for states
of $n$-qubits.
\end{conjecture}

\subsection{\texorpdfstring{Bounds on $A_n$}{Bounds on An}}\label{sec:boundsan}

Finally, we look at the full-body correlations of states, i.e.~$A_{n}$
of an $n$-qubit state. Lower bounds on this quantity can be used
to detect entanglement \cite{klockl2015characterizing,tran2016correlations}.
Upper bounds, however, are so far only known for the case of odd $n$
\cite{tran2016correlations}. In that case, combining again the purity $M_0 = 0$ from (\ref{eq:schmidt}) and state inversion inequality $B_0 \geq 0$ from (\ref{eq:shadowineq}) yields for odd $n$
\begin{eqnarray}
2^{n-1}\geq\!\!\!\!\!\!\sum_{k\text{~odd},~k\leq n}\!\!\!\!\!\!A_{k}\geq A_{n}.
\end{eqnarray}

For example, the $n$-partite GHZ state for odd $n$
fulfills $A_{n}=2^{n-1}$, thus this bound is tight.

For $n$ even, this trick does not work. In this case, the GHZ state
fulfills $A_{n}=2^{n-1}+1$, which is why it was conjectured in \cite{tran2016correlations}
that this is the upper bound. Here, we show that this is true at least
up to $n=10$.

For small $n$, this follows from the shadow inequality $B_{1}$ in (\ref{eq:shadowineq}).
Evaluating $B_{1}\geq 0$ for $n=2$ yields
\begin{eqnarray}
A_{2}\leq3=2^{2-1}+1,
\end{eqnarray}
which is the well known bound on the two-body correlations in two-qubit
states and is compatible with the conjecture. For $n=4$, $B_{1}\geq0$
yields 
\begin{eqnarray}
A_{4}\leq3-2A_{1}+A_{2}\leq3+\binom{4}{2}=2^{4-1}=9,
\end{eqnarray}
where we used the result of Proposition~\ref{prop:c1}. For higher $n$,
we observe that for every state $\rho$, there exists another state
$\hat{\rho}=\frac{1}{2}(\rho+Y^{\otimes n}\rho^{\text{T}}Y^{\otimes n})$
with the same even correlations $P_{2k}$ and vanishing odd correlations
$P_{2k+1}$ \cite{Wyderka2018}. Thus, the bounds on an even sector length can be obtained
by setting w.l.o.g.~the odd correlations to zero, i.e.~$A_{2k+1}=0$.

For $n=6$, we investigate $B_{1}\geq 0$ and $B_{3}\geq 0$ and combine them to
eliminate $A_{4}$. This yields, using Proposition~\ref{prop:c1} again,
\begin{eqnarray}
A_{6}\leq18+A_{2}\leq33.
\end{eqnarray}
For $n=8$, we combine $B_{1}$, $B_{3}$ and $B_{5}$ to yield the
bound, for $n=10$ we combine $B_k$ for $k=1,3,5,7$:

\begin{table}[t]

\caption{Translation of the complete sets of sector bounds of two- and three-qubit states into
linear entropy and mutual entropy inequalities. The trivial bounds
$A_{j}\geq0$ are omitted. The translation among the representations
is given by (\ref{eq:AtoS}) - (\ref{eq:StoI}). The constraints are due to purity, state inversion [$B_0\geq0$ from (\ref{eq:shadowineq})], Schmidt decomposition [(\ref{eq:schmidt})] and symmetric strong subadditivity (SSSA, Thm.~\ref{thm:ssa}). \label{tab:constraints}}

\begin{tabular}{c|cccc}
\br
$n$ & Constraint & Sector length & Linear entropy & Mutual entropy\tabularnewline
\mr
2 & Purity & $A_{1}+A_{2}\leq3$ & $S_\text{L}^{(2)}\geq0$ & $I_\text{L}^{(2)}\leq I_\text{L}^{(1)}$\tabularnewline

2 & State inv. & $A_{1}-A_{2}\leq1$ & $S_\text{L}^{(2)}\leq S_\text{L}^{(1)}$ & $I_\text{L}^{(2)}\geq0$\tabularnewline
\mr
3 & Purity & $A_{1}+A_{2} + A_{3} \leq7$ & $S_\text{L}^{(3)}\geq0$ & $I_\text{L}^{(3)}\geq I_\text{L}^{(2)} -I_\text{L}^{(1)}$\tabularnewline

3 & State inv. &$A_{1}-A_{2}+A_{3} \leq1$ & $S_\text{L}^{(3)}\geq S_\text{L}^{(2)} -S_\text{L}^{(1)}$ &  $I_\text{L}^{(3)}\geq0$\tabularnewline

3 & Schmidt dec. & $A_{2}\leq3$ & $S_\text{L}^{(2)}\geq S_\text{L}^{(1)}$ & $I_\text{L}^{(2)}\leq I_\text{L}^{(1)}$\tabularnewline

3 & SSSA & $A_{1}+A_{2}\leq3(1 + A_{3})$ & $3S_\text{L}^{(3)}\leq 2S_\text{L}^{(2)} -S_\text{L}^{(1)}$ & $I_\text{L}^{(3)}\leq\frac{1}{3}I_\text{L}^{(2)}$\tabularnewline
\br
\end{tabular}

\end{table}

\begin{thm}
For $n$-qubit states with $n\leq10$, $n$ even, it holds that $A_{n}\leq2^{n-1}+1$.
The bound is tight.
\end{thm}
If Conjecture~\ref{conj:conj1} is true for $k=4$ and $n_{0}\leq12$,
as numerical calculation indicates, then the same method works for $n=12$, $n=14$,
$n=16$ as well.

\subsection{Application to entanglement detection}

Before continuing, we highlight some applications of the bounds found in this section to the detection of entanglement. As mentioned before, sector lengths are convex and invariant under local unitaries, making them useful for entanglement detection \cite{klockl2015characterizing}.
This can be exploited by noticing that for product states $\rho = \rho_A \otimes \rho_B$, where $\rho_A$ consists of $n_A$ and $\rho_B$ of $n_B$ particles, it holds that
\begin{eqnarray}\label{eq:productstates}
    A_k(\rho_A \otimes \rho_B) = \sum_{j=0}^{k} A_j(\rho_A)A_{k-j}(\rho_B),
\end{eqnarray}
where we set $A_k(\rho) = 0$ if $k$ exceeds the number of particles in $\rho$.

For $n_A = n_B = 1$, $A_{2}(\rho_{A} \otimes \rho_{B})=A_{1}(\rho_{A})A_{1}(\rho_{B})$. Due to convexity of the sector lengths, it follows that $A_{2}\leq1$ for all separable states.

For more than two parties, different entanglement structures occur. A multi-partite state $\rho$ is said to be biseparable, iff it can be written as
\begin{eqnarray}
    \rho = \sum_i p_i \rho_{A_i} \otimes \rho_{B_i},
\end{eqnarray}
where $\sum_i p_i = 1$ and the $A_i,B_i$ denote some bipartition of the parties, i.e. $A_i \dot\cup B_i = \{1,\ldots,n\}$.
A state is called genuinely multipartite entangled (GME), iff it is not biseparable.

For $n=3$, we showed that $A_{3}\leq4$, on the other hand, all biseparable
states obey $A_{3}\leq3$, as for states $\rho=\rho_{A}\otimes\rho_{BC}$
it holds that $A_{3}(\rho)=A_{1}(\rho_{A})A_{2}(\rho_{BC})\leq3$.
Therefore, also in this case, the highest sector length can be used
to detect genuine multipartite entanglement.

For $n=4$, however, the situation is different: One can show with the same argument as above that
biseparable states fulfill $A_{4}\leq9$. But, as seen before, $A_{4}\leq9$ is already the bound for
all states. Thus, $A_{4}$ does not allow for detection of genuine
multipartite entanglement. However, there is a nontrivial bi-separability
bound on $A_{3}$ of $7$, whereas  the bound of Proposition~\ref{prop:a3bound} due to positivity of the state is given by $A_{3}\leq8$. Therefore, not the highest, but the next-to-highest
correlations allow for entanglement detection. This already yields an entanglement criterion which can detect states not detectable by known criteria using the sector lengths \cite{klockl2015characterizing}, an example being again the highly entangled state from Eq.~(\ref{eq:5chi}) with sector length configuration $(A_1,A_2,A_3,A_4) = (0,2,8,5)$.
Note that it is known that even vanishing highest order correlations do not exclude genuine multipartite entanglement
 \cite{kaszlikowski2008quantum,laskowski2012incompatible,tran2017genuine,klobus2018higher}.
Finally, let us note that while sector lengths are quadratic expressions in the quantum state, 
the additional knowledge of similar quantities of higher order, i.e.~higher moments, allows for more refined entanglement detection \cite{ketterer2019characterizing}.

\section{Bounds on linear combinations of sector lengths}

We now turn to the problem of finding bounds on linear combinations of sector lengths. This is related to the question of whether linear constraints are enough to fully characterize the set, meaning that the set of states forms a polytope in the sector length picture.
As mentioned before, sector lengths are in one-to-one correspondence with linear entropies and the mutual entropy for linear entropies. It turns out that some of the obtained inequalities are easier understood in the language of linear entropies.

\subsection{Translation into entropy inequalities}\label{sec:translation}

The linear entropy of a state $\rho$ is defined as
$S_\text{L}(\rho)=2[1-\trace(\rho^{2})]$. As $\trace(\rho^2)$, the purity of $\rho$, is up to a factor equal to the sum of all sector lengths of $\rho$, we can express $S_\text{L}$ in terms of sector lengths.
We define the sector entropy of sector $k$ by summing over all linear entropies of reduced states of $k$ particles, i.e.
\begin{eqnarray}
    S_\text{L}^{(k)}&:=\sum_{\substack{K\subset\{1,\ldots,n\}}{  \vert K\vert=k}}S_\text{L}(\rho_{K}) \nonumber \\
&=\frac{1}{2^{k-1}}\left[\binom{n}{k}2^{k}-\sum_{j=0}^{k}\binom{n-j}{k-j}A_{j}\right],
\end{eqnarray}
which can be inverted to yield
\begin{eqnarray}\label{eq:AtoS}
A_{k}=\binom{n}{k}-\sum_{j=1}^{k}(-1)^{k-j}2^{j-1}\binom{n-j}{k-j}S_\text{L}^{(j)}.
\end{eqnarray}
Furthermore, it will be useful to define the $k$-partite mutual linear entropy, 
\begin{eqnarray}\label{eq:StoI}
    I_\text{L}^{(k)} := \sum_{j=1}^k (-1)^{j-1} \binom{n-j}{k-j}S_\text{L}^{(j)}.
\end{eqnarray}
For $k=2$ and $n=2$, it resembles the usual mutual entropy, $I_\text{L}^{(2)} = S_\text{L}(\rho_A) + S_\text{L}(\rho_B) - S_\text{L}(\rho_{AB})$.
Note that the definition is analogous to the mutual information of von~Neumann-entropy. However, in the case of linear entropy, the name mutual linear entropy is preferred, as the quantity is not additive and does not vanish for product states \cite{gour2007dual}. 
Table~\ref{tab:translation} lists the non-trivial bounds on the sector lengths found above, translated into the two other representations.

\begin{figure}[t]
    \centering
    \begin{tikzpicture}
    \begin{axis}[xmax=3.2,ymax=3.2,ymin=-0.2, samples=50,grid=both, xlabel=$A_1$, ylabel=$A_2$]
      \addplot[name path=f, blue, ultra thick][domain=0:2] (x,3-x);
      \node[rotate=-42] at (axis cs: 1.1,2.1) {purity bound};
    
      \path[name path=axis] (axis cs:0,0) -- (axis cs:1,0) -- (axis cs:2,1);
      \addplot [thick, color=blue, fill=blue, fill opacity=0.05]
        fill between[of=f and axis];
    
      \addplot[name path=sep, blue, ultra thick, dashed][domain=0:2] (x,1);
      \addplot[name path=sep, blue, thick, dotted][domain=1:2] (x,x-1);
      \node[rotate=42] at (axis cs: 1.4,0.55) {state inv.};
      \node[] at (axis cs: 0.7,1.5) {entangled};
      \node[right] at (axis cs: 2,1) {$\ket{00}$};
      \node[right] at (axis cs: 0,3) {$\,\,\,\ket{\Phi^+}$};
      \node[right] at (axis cs: 1,0) {$\,\,\,\frac12(\ketbra{00}{00}+\ketbra{01}{01})$};
    \end{axis}      
    \end{tikzpicture} 
    \caption{The total set of attainable pairs $A_1$ and $A_2$ in two-qubit states, displayed in light blue. }
    \label{fig:twoqubits}
\end{figure}
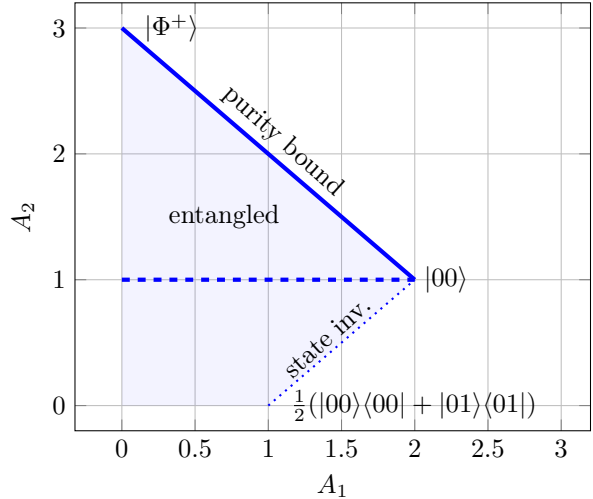

\subsection{Characterization of two- and three-qubit states}

Using the results above, we can now characterize the allowed values of sector length 
tuples $(A_1,\ldots,A_n)$ for two-qubit and three-qubit states. It turns out that in both
cases the set of admissible values is a convex polytope. This is interesting as 
the convexity is not trivial, because the sector lengths are nonlinear in the state. 
In addition, it is surprising that only a finite number of linear constraints 
corresponding to the surfaces of the polytope is sufficient for a full description.
This reminds of a similar polytope for separable states, if variances of collective
spin-observables are considered \cite{toth2009spin}.

\subsubsection{The case of two qubits}\label{sec:twoqubits}

It is easy to verify that in the case of $n=2$, pure product states obey $A_{1}=2$ and $A_{2}=1$ [see (\ref{eq:productstates})]. The Bell state $\ket{\Phi^+} = \frac{1}{\sqrt2} (\ket{00} + \ket{11})$ obeys $A_{1}=0$, $A_{2}=3$. The purity bound $\trace(\rho^2) \leq 1$ translates into $1+A_{1}+A_{2}\leq4$. By superposing a pure product state and the Bell state, one can obtain
pure states with $A_{1}\in[0,2]$ and $A_{2}=3-A_{1}$. Exceeding
the value of 2 for $A_{1}$ is impossible due to the bound $A_1\leq n$ from Eq.~(\ref{eq:bounda1}).

However, the state inversion bound $B_0\geq0$ from (\ref{eq:shadowineq}) yields
another bound on $A_{1}$ and $A_{2}$ due to positivity; namely $A_{1}-A_{2}\leq1$.
States reaching this bound are given by the family $(1-p)\ketbra{00}{00}+p\ketbra{01}{01}$. All other states can be reached by mixing the boundary states with the maximally mixed state $\one/4$, as these states lie on a straight line connecting the boundary state with the origin. This yields the whole set of admissible pairs of $A_1$ and $A_2$ and is displayed in figure~\ref{fig:twoqubits}.

\begin{figure}[t]
    \centering
    \includegraphics[width=0.7\columnwidth]{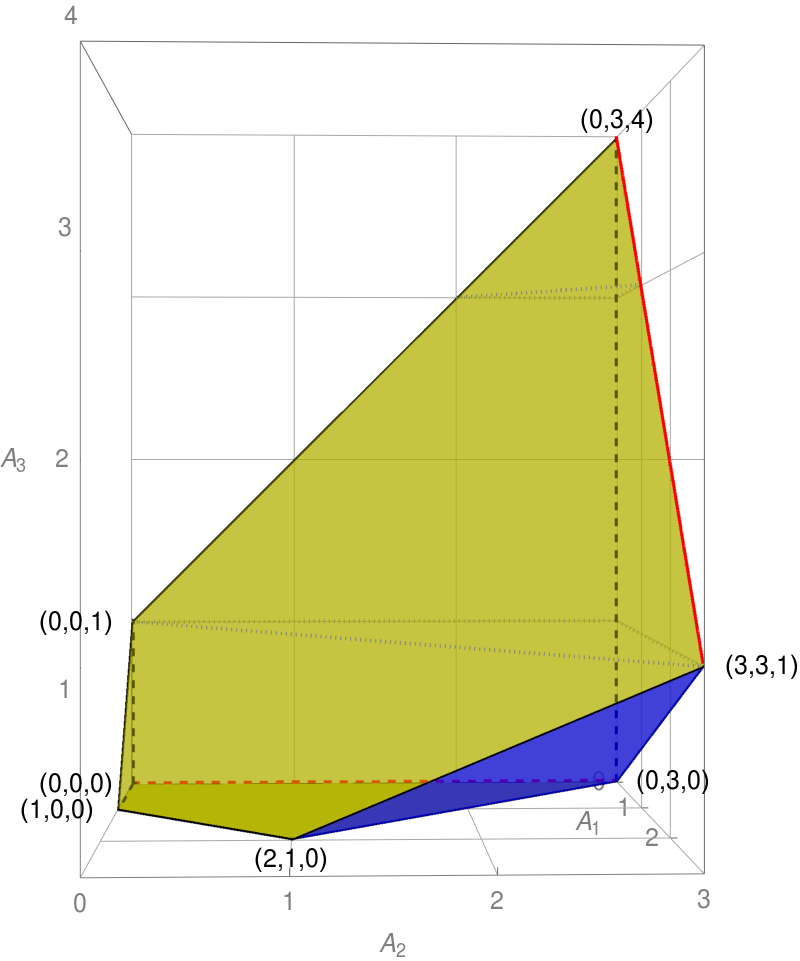}
    \caption{The polytope of admissible sector length configurations of three-qubit states. The yellow surface corresponds to the state inversion bound $B_0 \propto 1 - A_1 + A_2 - A_3 \geq 0$, the blue surface originates from symmetric strong subadditivity [(\ref{eq:sssa})]. Pure states lie on the red solid line connecting $(3,3,1)$ (product states) and $(0,3,4)$ (GHZ state). The $A_2$-axis is displayed by a red dashed line. 
    States above the the lower gray dotted line are not fully separable, states above the upper gray dotted line are genuinely multipartite entangled.
    The surface of the polytope is displayed in figure~\ref{fig:yellowsurface}.}
    \label{fig:threequbits}
\end{figure}

\subsubsection{The case of three qubits}\label{sec:threequbits}

While all the bounds in the case of two qubits are known, the case of three qubits shows an interesting new result that is connected to strong subadditivity of linear entropy. 

We start by collecting all inequalities that we know: The state inversion bound $B_0\geq 0$ from Eq.~(\ref{eq:shadowineq}), the bound $A_1 \leq 3$, the shadow inequality $B_1\geq 0$ and the bound from Proposition~\ref{prop:c1} yield a set of four inequalities,
\begin{eqnarray}
    1 - A_1 + A_2 - A_3 & \geq 0, \qquad A_1 \leq 3, \label{eq:poly1} \\
    9 - 5 A_1 + A_2 + 3A_3 & \geq 0, \qquad A_2 \leq 3, \label{eq:poly3}
\end{eqnarray}
from which the bound $\trace(\rho^2)\leq 1$ can be obtained using a linear program. These inequalities define a polytope in the three-dimensional space of tuples $(A_1,A_2,A_3)$.

However, as numerical search indicates, these bounds are not tight. As it turns out, there is a single additional linear constraint replacing the constraint $B_1\geq0$.

\begin{thm}\label{thm:ssa}
For $3$-qubit states, it holds that
\begin{eqnarray}
    A_1 + A_2 \leq 3(1+A_3). \label{eq:sssa}
\end{eqnarray}
\end{thm}
The proof is given in Appendix B and uses a semidefinite program for a relaxed version of the problem.
The polytope defined by (\ref{eq:poly1})-(\ref{eq:sssa}) is displayed in figure~\ref{fig:threequbits} and figure~\ref{fig:yellowsurface} in the Appendix.

It remains to show that the obtained polytope is tight by showing the existence of states for every point in the polytope, including the boundary. In fact, it suffices to find states on the yellow and the blue surface in figure~\ref{fig:threequbits}, corresponding to the state inversion bound $1-A_1+A_2-A_3 \geq 0$ and the bound $A_1 + A_2 \leq 3(1+A_3)$ from Theorem~\ref{thm:ssa}. This follows from the observation that for every state $\rho$, also the state inversion $\tilde{\rho} := Y^{\otimes n}\rho^{\text{T}} Y^{\otimes n}$ is a proper state, with the same coefficients in the Bloch decomposition up to a minus sign for all coefficients of an odd number of Pauli operators \cite{Wyderka2018}. Thus, the family $\rho(p) = p\rho + (1-p)\tilde\rho$ corresponds to states with sector lengths $((1-2p)^2 A_1(\rho), A_2(\rho), (1-2p)^2 A_3(\rho))$, yielding a family of states lying on a straight line connecting a point in the polytope with the point $(0, A_2, 0)$ on the red dashed $A_2$-axis with the same value of $A_2$. Therefore, states filling the yellow and the blue surface and their straight-line connections to the $A_2$-axis fill the whole polytope.

We find and list these boundary states explicitly in Appendix D, where we also display the net of the polytope.

\subsubsection{Connection to strong subadditivity}\label{sec:sssa}

Theorem~\ref{thm:ssa} is closely related to strong subadditivity (SSA). One formulation of SSA for the specially chosen particle $B$ is $S(\rho_{ABC}) + S(\rho_B) \leq S(\rho_{AB}) + S(\rho_{BC})$. However, it  holds for the von~Neumann entropy  only and fails to hold for the linear entropy, a counterexample being the state $\ket{\Phi^{+}}\bra{\Phi^{+}}\otimes\frac{\one}{2}$. Nevertheless, summing SSA over all particles to symmetrize it, yields
\begin{eqnarray}\fl
    \phantom{I}3S_\text{L}(\rho_{ABC}) + S_\text{L}(\rho_A) +  S_\text{L}(\rho_B) +  S_\text{L}(\rho_C) \leq 2[S_\text{L}(\rho_{AB})+S_\text{L}(\rho_{AC})+S_\text{L}(\rho_{BC})],
\end{eqnarray}
or in our language,
\begin{eqnarray}
    3S_\text{L}^{(3)} + S_\text{L}^{(1)} \leq 2S_\text{L}^{(2)}.
\end{eqnarray}
This is, using the correspondence (\ref{eq:AtoS}), equivalent to the statement of Theorem~\ref{thm:ssa}. Thus, linear entropy for three-qubit states obeys a symmetric SSA, which implies that usual SSA holds for at least one choice of special particle. Another formulation in terms of mutual linear entropies yields the inequality $I_\text{L}^{(3)} \leq \frac{1}{3}I_\text{L}^{(2)}$.

We state the full set of restrictions for $n=2$ and $n=3$ in all three representations in table~\ref{tab:constraints}.

Finally, note that the statement of Theorem~\ref{thm:ssa} can be generalized to states of more particles using the same induction trick as in the proof of Lemma~\ref{lem:l1}. We get:
\begin{cor}
For $n$-qubit states with $n\geq3$, it holds that $I_\text{L}^{(3)} \leq \frac{n-2}{3}I_\text{L}^{(2)}$.\label{prop:i3bound}
\end{cor}
\noindent In terms of sector lengths, the bound reads 
\begin{eqnarray}
    \binom{n}{3}-\frac13 \binom{n-1}{2} A_1 - \frac13 \binom{n-2}{1}A_2 + A_3 \geq 0.
\end{eqnarray}
Using a linear program, it is evident that this equation is stronger than the shadow inequalities (\ref{eq:shadowineq}).
As this bound is complementary to the bound $A_3\leq \binom{n}{3}$, we list it as well in table~\ref{tab:translation}.

\section{Conclusions}
We showed how to combine methods from quantum mechanics, coding theory and 
semidefinite programming to obtain strict bounds on linear combinations of 
sector lengths for multi-qubit systems. As a result, we obtained a full 
characterization of the allowed tuples of sector lengths for $n\leq 3$, 
where for $n=3$ one of the constraints is related to a symmetrized 
version of strong subadditivity of linear entropies. Our results can 
be understood in the language of entropy inequalities and monogamy 
relations, they can also be used  in the context of entanglement 
detection and the representability problem.

Our results highlight several problems for further research. First of all, 
the natural question of a complete characterization of sector bounds for 
$n \geq 4$, but also for higher-dimensional systems beyond qubits arises. The notion of sector lengths can be extended to higher-dimensional states as well, and many of the techniques like state inversion can be generalized. This has been used in the past to obtain some bounds \cite{eltschka2018distribution, eltschka2018exponentially}, however, a complete characterization is still out of reach.
Interestingly, we found that for $n\leq3$, the allowed region of sector bounds turned 
out to be a polytope, perfectly described by few linear constraints. The 
reason for this remains elusive and deserves further attention, as it may 
yield deep insight into the complicated structure of the positivity 
constraints. It might well be that this is a feature exclusive to 
qubit systems, or systems of few particles only. Apart from a similar 
characterization of higher-dimensional states of more parties, a deeper 
understanding of the associated entropy inequalities is crucial. For instance, 
the question of whether the inequality holds for other entropies is relevant. 

Finally, we used the bounds on the individual sector lengths for the task of entanglement detection and showed, for instance, that $A_4$ alone does not allow for entanglement detection in four-qubit states. The question of whether linear combinations of sector lengths allow for better entanglement criteria in this scenario is worth to be investigated in the future.

\ack
We thank Christopher Eltschka, Felix Huber, Jens Siewert, Timo Simnacher,
and Zhen-Peng Xu for useful discussions. This work has been supported by 
the DFG, the ERC (Consolidator Grant No. 683107/TempoQ), and the House of 
Young Talents of the Universität Siegen.

\appendix
\setcounter{section}{1}

\section{Proof of Proposition \ref{prop:c2}}

In this section, we prove Proposition \ref{prop:c2} from the main text:

\setcounterref{definition}{prop:c2}
\addtocounter{definition}{-1}
\begin{prop}
For all qubit states of $n$ parties with $n \geq 4$, it holds that
$\sum_{j=2}^{n}A_{2}(\rho_{1j})\leq n-1$.
\end{prop}
\begin{proof}
We prove the claim for $n=4$ first. In this case we distribute all Pauli operators whose expectation values contribute to the bipartite sector lengths into anticommuting sets,
\begin{eqnarray*}\fl
\phantom{Ig}M_{1} & =  \{\X\X\one \one ,\X\Y\one \one ,\X\Z\one \one , \: \Y\one \X\one ,\Y\one \Y\one ,\Y\one \Z\one , \: \Z\one \one \X,\Z\one \one \Y,\Z\one \one \Z\},\\ \fl
\phantom{Ig}M_{2} & =  \{\Y\X\one \one ,\Y\Y\one \one ,\Y\Z\one \one , \: \Z\one \X\one ,\Z\one \Y\one ,\Z\one \Z\one , \: \X\one \one \X,\X\one \one \Y,\X\one \one \Z\},\\ \fl
\phantom{Ig}M_{3} & =  \{\Z\X\one \one ,\Z\Y\one \one ,\Z\Z\one \one , \: \X\one \X\one ,\X\one \Y\one ,\X\one \Z\one , \: \Y\one \one \X,\Y\one \one \Y,\Y\one \one \Z\},
\end{eqnarray*}
such that in each set all operators pairwise anticommute. Here, $XX\one \one$ means again $X\otimes X \otimes \one \otimes \one$. For any anticommuting set $M$, it holds that $\sum_{m\in M} \langle m \rangle^2 \leq 1$ \cite{Kurzyski2011, toth2009spin, Wehner2008}. The sets are chosen such that
\begin{eqnarray}
\sum_{j=2}^{4}A_{2}(\rho_{1j}) = \sum_{i=1}^3 \sum_{m\in M_i}  \langle m \rangle^2 \leq 3.
\end{eqnarray}

To augment the proof to the case of $n>4$, we consider all $\binom{n-1}{3}$ subsets of four of the parties containing the first one, i.e., for $n=5$ we would consider the sets $\{1,2,3,4\}, \{1,2,3,5\}, \{1,2,4,5\}$ and $\{1,3,4,5\}$. For each of these subsets, the inequality for four parties holds. Summing these inequalities yields, on the one hand, an upper bound of $3\binom{n-1}{3}$. On the other hand, we obtain each of the two-body correlations $A_2({\rho_{1j}})$ exactly $\binom{n-2}{2}$ times. Dividing both sides by this factor proves the claim.
\end{proof}

\section{Proof of Theorem \ref{thm:ssa}}

In this section, we prove the symmetric strong subadditivity for three-qubit states.

\setcounterref{definition}{thm:ssa}
\addtocounter{definition}{-1}
\begin{thm}
For $3$-qubit states, it holds that $A_1 + A_2 \leq 3(1+A_3)$.
\end{thm}
\begin{proof}
Consider the map $\rho^\prime = M(\rho) := (YY\one)\rho^{\text{T}_{12}}(YY\one) + (Y\one Y)\rho^{\text{T}_{13}}(Y\one Y) + (\one YY)\rho^{\text{T}_{23}}(\one YY)$,
where $\rho^{\text{T}_{ij}}$ is the partial transpose of $\rho$ w.r.t.~systems $i$ and $j$. This map can be seen as a sum of partial state inversions of subsystems of size two, flipping the sign of the Pauli matrices of that particular subsystems. Using the Bloch decomposition, it can easily be seen that $\trace(\rho \rho^\prime) = \frac{1}{8}(3 - A_1 - A_2 + 3A_3)$. 
Note that the map defined above is not positive, however, we will show that $\trace(\rho \rho^\prime)\geq 0$ for all $\rho$, yielding the claim.

To that end, we consider the Choi matrix $\eta$ of the map, given by $\eta = (\one\otimes M)(\ketbra{\phi^+}{\phi^+})$ with $\ket{\phi^+} = \frac{1}{\sqrt{8}}\sum_{i=0}^7 \ket{ii}$ \cite{dePillis1967linear, choi1975completely}. The map can be reconstructed via $M(\rho) = 2^3\trace_A[(\rho^{\text{T}}\otimes \one)\eta]$. Thus, the quantity in question can be written in terms of the Choi matrix as $\trace(\rho \rho^\prime) = 2^3\trace[ (\rho \otimes \rho) \eta^{\text{T}_A}]$. As $M$ is not positive, $\eta$ is not positive as well, and one can directly calculate that $\eta^{\text{T}_A}$ has a single negative eigenvalue of $-3/2$. Nevertheless, it is positive for symmetric product states $\rho \otimes \rho$. To see this, we use an SDP to minimize $\trace(\sigma \eta^{\text{T}_A})$ over symmetric states $\sigma$ and trying to enforce the product structure on $\sigma$ using some relaxations of this property.

To begin with, the matrix $\eta^{\text{T}_A}$ can be written in Bloch decomposition as
\begin{eqnarray}\fl
     \eta^{\text{T}_A} \propto 3\one\one\one\;\one\one\one  - \!\!\!\!\sum_{a\in\{x,y,z\}}\!\!\!\! \sigma_a \one \one \;  \sigma_a \one \one  - \!\!\!\!\sum_{a\in\{x,y,z\}}\!\!\!\! \one \sigma_a \one \; \one \sigma_a \one  - \!\!\!\!\sum_{a\in\{x,y,z\}}\!\!\!\! \one \one \sigma_a \; \one \one \sigma_a \nonumber \\ \fl \qquad\qquad
      - \!\!\!\!\sum_{a,b\in\{x,y,z\}}\!\!\!\! \sigma_a \sigma_b \one \;  \sigma_a \sigma_b \one  - \!\!\!\!\sum_{a,b\in\{x,y,z\}}\!\!\!\! \sigma_a \one \sigma_b \; \sigma_a \one \sigma_b  - \!\!\!\!\sum_{a,b\in\{x,y,z\}}\!\!\!\! \one \sigma_a \sigma_b \; \one \sigma_a \sigma_b \nonumber \\ \fl \qquad\qquad
      +3\!\!\!\!\sum_{a,b,c\in\{x,y,z\}} \!\!\!\!\sigma_a \sigma_b \sigma_c \; \sigma_a \sigma_b \sigma_c.
\end{eqnarray}
Note that due to the special symmetric form of the basis elements, the matrix can also be written as a combination of local flip operators. This allows to write the matrix also in terms of projectors onto the symmetric and antisymmetric subspaces. This representation of the problem is explained in more detail in Appendix C.

The matrix $\eta^{\text{T}_A}$ exhibits many symmetries; it is symmetric under the exchange of the first three and the second three parties. Also, it is symmetric under any permutation among the first three parties, if the same permutation is applied to the second three parties as well.  Furthermore, it is invariant under single qubit local unitaries $V\one\one V\one\one$ for $V \in \{X,Y,Z,\Pi, T, H\}$ where $\Pi = \diag(1,i)$, $T=\diag(1,\exp(i\pi/4))$ and $H$ being the Hadamard gate.

All of these symmetries do not alter the product structure of $\rho \otimes \rho$ and can therefore be imposed for the optimal state as well. 

Apart from the symmetries, we can try to impose the product structure of $\sigma$. However, this is a non-linear constraint and thus not exactly tractable by an SDP. Nevertheless, we find a set of linear constraints that brings us close enough to the set of product states to prove the claim.

First of all, product states are separable by definition and must have a positive partial transpose, i.e. $\sigma^{\text{T}_A} \geq 0$ \cite{Peres1996}.
Next, using the positivity of Breuer-Hall maps, for separable states $\sigma$ and skew symmetric unitaries $U$, i.e., $U^\text{T} = -U$, it holds that
$\sigma_{\text{BH}} = \trace_{4,5,6}(\sigma) \otimes \one\one\one - \sigma - (\one\one\one \otimes U)\sigma^{\text{T}_B}(\one\one\one \otimes U^\dagger) \geq 0$ \cite{Breuer2006, Hall2006}. It turns out that the choice of $U=YYY$ is suitable in our case.

As a last constraint, for product states, $\langle A \otimes A \rangle_{\rho \otimes \rho} = \langle A \rangle_\rho^2 \geq 0$ for all three-qubit observables $A$. Here, we consider the special choice of $A=X\one\one$. For product states, it should hold that $\langle A\otimes A\rangle_\sigma = \langle A\otimes \one\one\one \rangle_\sigma^2$, as $\sigma$ is symmetric as noted before. To make this constraint linear, note that for Pauli observables, $\vert \langle A \rangle \vert \leq 1$. Thus, $\langle A\otimes A\rangle_\sigma \leq \vert \langle A\otimes \one\one\one \rangle_\sigma \vert$. Now, there are two possibilities. Either, the optimal state obeys $\langle A\otimes \one\one\one \rangle_\sigma \geq 0$ or $\langle A\otimes \one\one\one \rangle_\sigma \leq 0$. Therefore, we run the SDP twice, once with the constraint $\langle A\otimes A\rangle_\sigma \leq \langle A\otimes \one\one\one \rangle_\sigma$ and once with $\langle A\otimes A\rangle_\sigma \leq -\langle A\otimes \one\one\one \rangle_\sigma$.

To summarize, we run the following SDP:
\begin{eqnarray} \fl
    \phantom{Ig}\underset{\sigma}{\text{min}} & \trace(\sigma \eta^{\text{T}_A}) \\ \fl
    \phantom{Ig}\text{subject~to~}\quad & \sigma \geq 0, \\
                            & \sigma \text{~symmetric}, \\
                            & (V\one\one\,V\one\one) \sigma (V\one\one\,V\one\one) = \sigma \text{~for~} V\in\{X,Y,Z,\Pi,T,H\}, \\
                            & \sigma^{\text{T}_A}\geq 0,~ \sigma_{\text{BH}}\geq0,\\
                            & \trace[(A\otimes A)\sigma] \geq 0 \text{~for~all~observables~}A, \\
                            & \trace[(X\one\one\,X\one\one) \sigma] \leq \pm \trace[(X\one\one\,\one\one\one) \sigma ].
\end{eqnarray}

Here, the symmetry constraint means both, symmetric under exchange of the first three with the last three parties, as well as symmetric under exchange among the first three and the same exchange among the last three parties. The last three constraints are the linear approximations of the product structure, where the $\pm$ in the last constraint means that we run the SDP once for each choice.
Both cases yield a minimal trace of zero, proving the claim.
\end{proof}

The method presented here can also be used to prove bounds for arbitrary linear combinations $\sum_k c_k A_k$. In this case, one has to choose $\eta^{\text{T}_A} = \sum_k c_k \sum_{\Xi_k} \Xi_k \otimes \Xi_k$, where the inner sum iterates over all Pauli operators $\Xi_k$ acting on $k$ of the parties nontrivially, as well as choosing appropriate relaxations of the product structure.

\section{Representation using symmetric and antisymmetric subspaces}

As noted in Appendix B, finding bounds on linear combinations of sector lengths is equivalent to solving a quadratic program to find $\min_\rho \trace[(\rho^{(A)} \otimes \rho^{(B)}) \eta]$
with $\eta = \sum_k c_k \sum_{\Xi_k}  \Xi_k^{(A)} \otimes \Xi_k^{(B)}$. Due to the special symmetric form, it is possible to express $\eta$ in terms of local flip operators $F=\frac12\sum_{j=0,x,y,z} \sigma_j \otimes \sigma_j$, which in turn can be written in their eigenbasis with the eigenvectors given by the projectors $\Pi_- = \ketbra{\Psi^-}{\Psi^-}$ and $\Pi_+ = \ketbra{\Psi^+}{\Psi^+} + \ketbra{\Phi^-}{\Phi^-} +\ketbra{\Phi^+}{\Phi^+}$ onto the antisymmetric and symmetric subspace, respectively. Here, $\ket{\Psi^\pm}$ and $\ket{\Phi^\pm}$ denote the usual Bell states. In this representation, the linear combination of sector lengths can be expressed as
\begin{eqnarray}
    \eta = \sum_{i_1\ldots i_n=\pm} \tilde c_{i_1\ldots i_n} \Pi_{i_1\ldots i_n}
\end{eqnarray}
with $\Pi_{i_1\ldots i_n} = \Pi_{i_1}^{(A_1,B_1)} \otimes \ldots \otimes \Pi_{i_n}^{(A_n,B_n)}$. The prefactors $\tilde{c}$ are connect to the prefactors $c$. This representation was used before to find entanglement witnesses and monotones \cite{mintert2005concurrence, mintert2007observable, demkowicz2006evaluable}. In these references, the authors restrict themselves to $\tilde{c}_{i_1\ldots i_n} \geq 0$ to ensure positivity. As we have seen in Appendix B, this approach is too restrictive, as positivity under trace with symmetric product states is not equivalent to positivity of the matrix $\eta$. Nevertheless, it is interesting to note that the relevant inequalities in the case of three qubits have a particular form in this representation. The matrix $\eta$ that yields the symmetric strong subadditivity is obtained by choosing $\tilde{c}_{---} = -3, \tilde{c}_{--+} = \tilde{c}_{-+-} = \tilde{c}_{+--} = 1$ and all other prefactors vanishing. The constraint $A_2 \leq 3$, however, can be expressed by choosing $\tilde{c}_{---} = -3, \tilde{c}_{-++} = \tilde{c}_{+-+} = \tilde{c}_{++-} = 1$. Usual state inversion is represented by $\tilde{c}_{---} = 1$. Therefore, it seems that the relevant inequalities correspond to some sort of extremal points in the set of coefficients $\tilde{c}$ that yield matrices that are positive under trace with positive product operators.

\section{Three-qubit states spanning the whole sector length space}

In this section, we explicitly state families of states that cover the whole three-qubit sector-length space displayed in figure~\ref{fig:threequbits}. First, we give families of states covering the yellow surface displayed in figure~\ref{fig:yellowsurface}, corresponding to the state inversion bound:
\begin{eqnarray}
    \fl \phantom{Ig}\rho_A(p,\alpha) & = p\ketbra{G(\alpha)}{G(\alpha)} + \frac{1-p}{8}(\one + XXX), \label{eq:fam1} \\
    \fl \phantom{Ig}\rho_B(p,\alpha) & = p\ketbra{H(\alpha)}{H(\alpha)}  + \frac{1-p}{8}(\one + \cos(\alpha)Z\one\one + \sin(\alpha) XXX), \\
    \fl \phantom{Ig}\rho_C(p,q) & = \frac{p}2 \one\otimes \ketbra{00}{00} + \frac{q}2 \one \otimes \ketbra{01}{01} + (1-p-q) \ketbra{000}{000}, \label{eq:fam3}
\end{eqnarray}
with the abbreviations
\begin{eqnarray}
    \fl \phantom{Ig}\ket{G(\alpha)} & = \frac{1}{\sqrt{1 + \cos(\frac\alpha2)\sin(\frac\alpha2)}} \left[\cos(\frac\alpha2) \ket{\text{GHZ}} + \sin(\frac\alpha2)\ket{+\!+\!+}\right],\\
    \fl \phantom{Ig}\ket{H(\alpha)} & = \left[\cos(\frac\alpha2)\ket{0} -  \sin(\frac\alpha2)\ket{1}\right]\otimes \ket{+-},
\end{eqnarray}
where $0\leq p\leq 1$, $0\leq q \leq p$ and $0\leq \alpha \leq \pi$.

Second, the blue surface corresponding to symmetric strong subadditivity is spanned by the states
\begin{eqnarray}
    \fl \phantom{Ig}\rho_D(\alpha,\beta) = \frac{p}{2}\ketbra{\Phi(\alpha)}{\Phi(\alpha)} \otimes \one + \frac{1-p}{2} \ketbra{\Phi(\beta)}{\Phi(\beta )} \otimes \one \label{eq:fam4}
\end{eqnarray}
where $\ket{\Phi(\alpha)} = \cos(\alpha/2)\ket{00} + \sin(\alpha/2)\ket{11}$ and $p = \sin(\beta)/[\sin(\alpha) + \sin(\beta)]$. The angles $\alpha$ and $\beta$ take arbitrary values between $0$ and $\pi$.

All other states can be reached by mixing these states with their inverted states, defined by $\tilde{\rho} := Y^{\otimes n} \rho^{\text{T}} Y^{\otimes n}$.
\pagebreak
\noappendix
\setcounter{figure}{2}
\begin{figure}[h!]
    \centering
    \includegraphics[width=1\columnwidth]{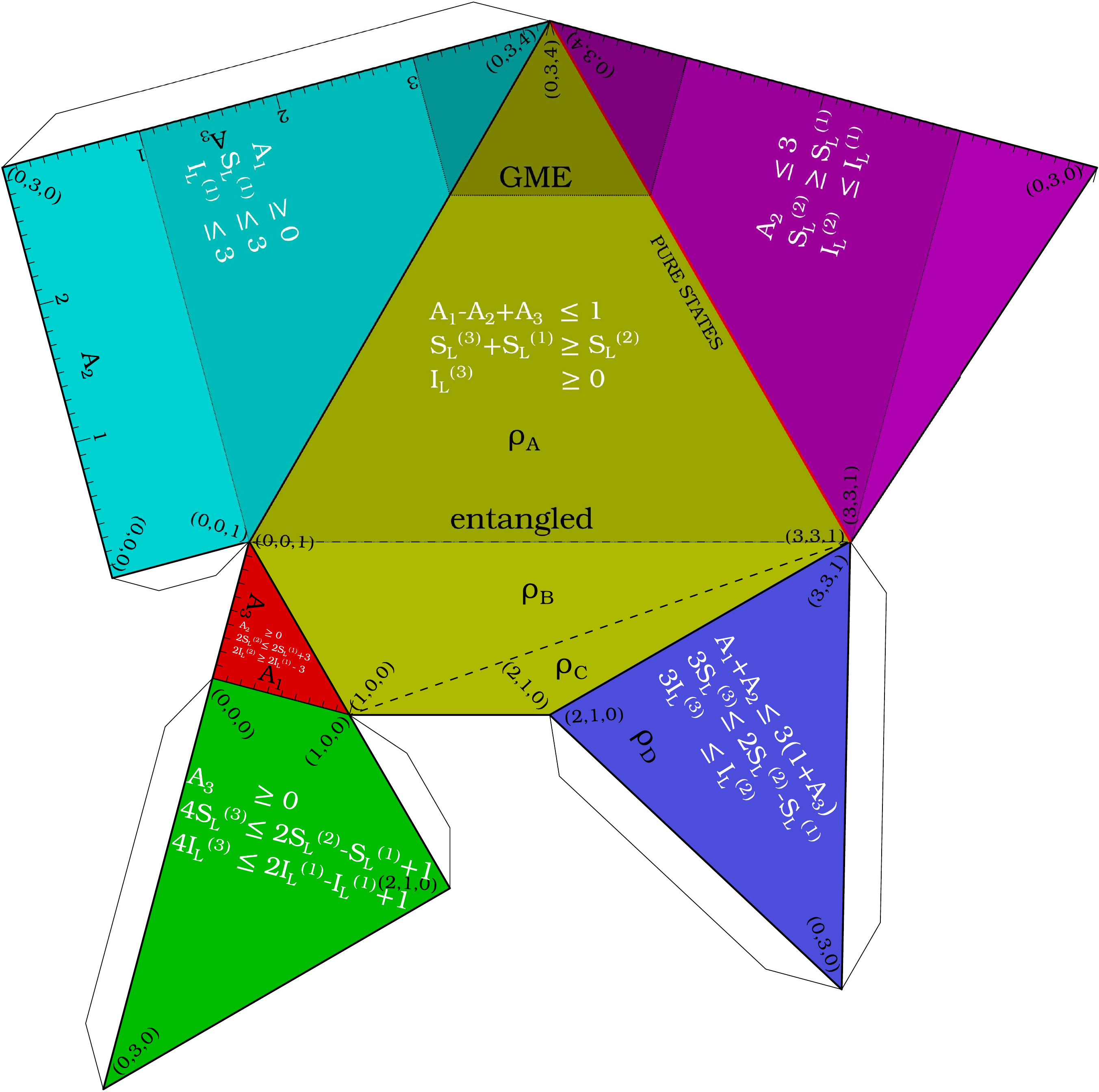}
    \caption{The surface of the polytope of admissible sector lengths of three-qubit states with the regions covered by the families of states (\ref{eq:fam1}) - (\ref{eq:fam3}) on the yellow surface corresponding to the state inversion bound $A_1-A_2+A_3=1$. The family of states (\ref{eq:fam4}) covers the whole of the blue surface corresponding to symmetric strong subadditivity $A_1+A_2 \leq 3(1+A_3)$. }
    \label{fig:yellowsurface}
\end{figure}
\pagebreak
\appendix
\newcommand{\newblock}{}
\bibliographystyle{apsrev4-1}
\bibliography{cite}

\end{document}